\documentclass[11pt]{article}
\usepackage{hyperref}
\usepackage{amssymb,amsmath,amsthm,mathrsfs}
\usepackage{epsfig}
\usepackage{times}

\usepackage[left=1.17in,right=1.17in,top=1in,bottom=1.17in]{geometry}

\newcommand{\R}{\mathbb{R}}

\newcommand{\MaxCSP}{\textsc{Max\ CSP}}

\newcommand{\vect}[1]{\mathbf{#1}}

\DeclareMathOperator*{\E}{\mathbb{E}}

\DeclareMathOperator{\Supp}{Supp}

\theoremstyle{plain}
\newtheorem{theorem}{Theorem}[section]
\newtheorem{lemma}[theorem]{Lemma}

\theoremstyle{definition}
\newtheorem{definition}[theorem]{Definition}

\renewcommand{\epsilon}{\varepsilon}

\newcommand{\soundness}{s}
\newcommand{\completeness}{c}

\newcommand{\SubHull}{\textsf{SM}}
\newcommand{\SubHullSym}{\textsf{SM}^{\textsf{sym}}}

\newcommand{\C}{\mathcal{C}}
\newcommand{\Csym}{\mathcal{C}^{\textsf{sym}}}
\newcommand{\fsym}{f^{\textsf{sym}}}

\begin{document}
\title{Improved Inapproximability For Submodular Maximization}

\author{ Per Austrin\thanks{ Supported by NSF Expeditions grant CCF-0832795.}\\
  New York University
}

\maketitle

\begin{abstract}
  We show that it is Unique Games-hard to approximate the maximum of a
  submodular function to within a factor $0.695$, and that it is
  Unique Games-hard to approximate the maximum of a symmetric
  submodular function to within a factor $0.739$.  These results
  slightly improve previous results by Feige, Mirrokni and Vondr{\'a}k
  (FOCS 2007) who showed that these problems are NP-hard to
  approximate to within $3/4+\epsilon \approx 0.750$ and $5/6+\epsilon
  \approx 0.833$, respectively.
\end{abstract}


\section{Introduction}

Given a ground set $U$, consider the problem of finding a set $S
\subseteq U$ which maximizes some function $f: 2^U \rightarrow \R^+$
which is \emph{submodular}, i.e., satisfies
$$
f(S \cup T) + f(S \cap T) \le f(S) + f(T).
$$
for every $S, T \subseteq U$.  The submodularity property is also
known as the property of \emph{diminishing returns}, since it is
equivalent with requiring that, for every $S \subset T \subseteq U$
and $i \in U \setminus T$, it holds that
$$
f(T \cup \{i\}) - f(T) \le f(S \cup \{i\}) - f(S).
$$
There has been a lot of attention on various submodular optimization
problems throughout the years (e.g.,
\cite{nemhauser78analysis,lovasz83submodular,feige07maximizing}, see
also the first chapter of \cite{vondrak07submodularity} for a more
thorough introduction).  Many natural problems can be cast in this
general form -- examples include natural graph problems such as
maximum cut, and many types of combinatorial auctions and allocation
problems.

A further restriction which is also very natural to study is
\emph{symmetric} submodular functions.  These are functions which
satisfy $f(S) = f(\overline{S})$ for every $S \subseteq U$, i.e., a
set and its complement always have the same value.  A well-studied
example of a symmetric submodular maximization problem is the problem
to find a maximum cut in a graph.

Since it includes familiar NP-hard problems such as maximum cut as a
special case, submodular maximization is in general NP-hard, even in
the symmetric case.  As a side note, a fundamental and somewhat
surprising result is that submodular \emph{minimization} has a
polynomial time algorithm \cite{grotschel81ellipsoid}.

To cope with this hardness, there has been much focus on efficiently
finding good approximate solutions.  We say that an algorithm is an
$\alpha$-approximation algorithm if it is guaranteed to output a set
$S$ for which $f(S) \ge \alpha \cdot f(S_\text{OPT})$ where
$S_{\text{OPT}}$ is an optimal set.  We also allow randomized
algorithms in which case we only require that the expectation of
$f(S)$ (over the random choices of the algorithm) is at least $\alpha
\cdot f(S_\text{OPT})$.

In many special cases such as the maximum cut problem, it is very easy
to design a constant factor approximation (in the case of maximum cut
it is easy to see that a random cut is a $1/2$-approximation).  For
the general case of an arbitrary submodular functions, Feige et
al.~\cite{feige07maximizing} gave a $(2/5-o(1))$-approximation
algorithm based on local search, and proved that a uniformly random
set is a $1/2$-approximation for the symmetric case.  The
$(2/5-o(1))$-approximation has been slightly improved by Vondr{\'a}k
\cite{vondrak10submodular} who achieved a $0.41$-approximation
algorithm, which is currently the best algorithm we are aware of.

Furthermore, \cite{feige07maximizing} proved that in the \emph{(value)
  oracle model} (where the submodular function to be maximized is
given as a black box), no algorithm can achieve a ratio better than
$1/2+\epsilon$, even in the symmetric case.  However, this result says
nothing about the case when one is given an \emph{explicit
  representation} of the submodular function -- say, a graph in which
one wants to find a maximum cut.  Indeed, in the case of maximum cut
there is in fact a $0.878$-approximation algorithm, as given by a
famous result of Goemans and Williamson \cite{goemans95improved}.
In the explicit representation model, the best current hardness
results, also given by \cite{feige07maximizing}, are that it is
NP-hard to approximate the maximum of a submodular function to within
$3/4+\epsilon$ in the general case and $5/6+\epsilon$ in the symmetric
case.  

\subsection{Our Results}

In this paper we slightly improve the inapproximability results of
\cite{feige07maximizing}.  However, as opposed to
\cite{feige07maximizing} we do not obtain NP-hardness but only
hardness assuming Khot's \emph{Unique Games Conjecture} (UGC)
\cite{khot02power}.  The conjecture asserts that a problem known as
Unique Games, or Unique Label Cover, is very hard to approximate.  See
e.g.\ \cite{khot02power} for more details.  While the status of the
UGC is quite open, our results still imply that obtaining efficient
algorithms that beat our bounds would require a fundamental
breakthrough.

For general submodular functions we prove the following theorem.

\begin{theorem}
  \label{thm:main1}
  It is UG-hard to approximate the maximum of a submodular function to
  within a factor $0.695$.
\end{theorem}

In the case of symmetric functions we obtain the following bound.

\begin{theorem}
  \label{thm:main2}
  For every $\epsilon > 0 $ it is UG-hard to approximate the maximum
  of a symmetric submodular function to within a factor $709/960 +
  \epsilon < 0.739$
\end{theorem}

These improved inapproximability results still fall short of coming
close to the $1/2$-barrier in the oracle model.  Unfortunately, while
marginal improvments of our results may be possible, we do not believe
that our approach can come close to a factor $1/2$.  It remains a
challenging and interesting open question to determine the exact
approximability of explicitly represented submodular functions.

\subsection{Our Approach}
\label{sec:techniques}

As in \cite{feige07maximizing}, the starting point of our approach is
hardness of approximation for constraint satisfaction problems (CSPs),
an area which, due to much progress during the last 15 years, is today
quite well understood.  Here it is useful to take a slightly different
viewpoint.  Instead of thinking of the family of subsets $2^U$ of $U$,
we consider the set of binary strings $\{0,1\}^n$ of length $n = |U|$,
indentified with $2^U$ in the obvious way.  These views are of course
equivalent and throughout the paper we shift between them depending on
which view is the most convenient.

For a string $x \in \{0,1\}^n$ and a $k$-tuple $C \in [n]^k$ of
indices, let $x_C \in \{0,1\}^k$ denote the string of length $k$
which, in position $j \in [k]$ has the bit $x_{C_j}$.  Now, given a
function $f: \{0,1\}^k \rightarrow \R^+$, we define the problem
$\MaxCSP^+(f)$ as follows.  An instance of $\MaxCSP^+(f)$ consists of
a list of $k$-tuples of variables $C_1, \ldots, C_m \in [n]^k$.
These specify a function $F: \{0,1\}^n \rightarrow \R^+$ by
$$
F(x) = \frac{1}{m} \sum_{i=1}^m f(x_{C_i})
$$
and the problem is to find an $x \in \{0,1\}^n$ to maximize $x$.

Note that if $f$ is submodular then every instance $F$ of $\MaxCSP^+(f)$
is submodular and $\MaxCSP^+(f)$ is a special case of the submodular
maximization problem.  

Next, we use a variation of a result by the author and Mossel
\cite{austrin09approximation}.  The result of
\cite{austrin09approximation} is for CSPs where one allows
\emph{negated literals}\footnote{Where each ``constraint''
  $f(x_{C_i})$ of $F$ is of the more general form $f(x_{C_i} + l_i)$
  for some $l_i \in \{0,1\}^k$, where $+$ is interpreted as addition
  over $GF(2)^k$.}, which can not be allowed in the context of
submodular maximization.  However, in Theorem~\ref{thm:ughardness} we
give a simple analogue of the result of \cite{austrin09approximation}
for the $\MaxCSP^+(f)$ setting.

Roughly speaking the hardness result says the following.  Suppose that
there is a \emph{pairwise independent} distribution $\mu$ such that
the expectation of $f$ under $\mu$ is at least $c$, but that the
expectation of $f$ under the \emph{uniform} distribution is at most
$s$.  Then $\MaxCSP^+(f)$ is UG-hard to approximate to within a factor
of $s/c$.

The hardness result suggests the following natural approach: take a
pairwise independent distribution $\mu$ with small support, and let
$\vect{1}_\mu: \{0,1\}^k \rightarrow \{0,1\}$ be the indicator
function of the support of $\mu$.  Then take $f$ to be a ``minimum
submodular upper bound'' to $\vect{1}_\mu$, by which we mean a
submodular function satisfying $f(x) \ge \vect{1}_\mu(x)$ for every
$x$ while having small expectation under the uniform distribution.

To make this plan work, there are a few small technical complications
(hidden in the ``roughly speaking'' part of the description of the
hardness result above) that we need to overcome, making the final
construction slightly more complicated.  Unfortunately, understanding
the ``minimum submodular upper bound'' of the families of indicator
functions that we use appears difficult, and to obtain our results, we
resort to explicitly computing the resulting submodular functions for
small $k$.

Let us compare our approach with that of \cite{feige07maximizing}.  As
mentioned above, their starting point is also hardness of
approximation for constraint satisfaction.  However, here their
approach diverges from ours: they construct a gadget reduction from
the \textsc{$k$-Lin} problem (linear equations $\bmod$ $2$ where each
equation involves only $k$ variables).  This gadget introduces two
variables $x_i^0$ and $x_i^1$ for every variable $x_i$ in the
\textsc{$k$-Lin} instance, and each equation $x_{i_1} \oplus \ldots
\oplus x_{i_k} = b$ is replaced by some submodular function $f$ on the
$2k$ new variables corresponding to the $x_{i_j}$'s.  The analysis
then has to make sure that there is always an optimal assignment where
for each $i$ exactly one of $x_i^0$ and $x_i^1$ equals $1$, which for
the inapproximability of $3/4$ becomes quite delicate.  In our
approach, which we feel is more natural and direct, we don't run into
any such issues.

\subsection{Organization}

In Section~\ref{sect:notation} we set up some more notation that we
use throughout the paper and give some additional background.  In
Section~\ref{sect:hardness} we describe the hardness result that is
our starting point.  In Section~\ref{sect:construction} we describe in
more detail the construction outlined above, and finally, in
Section~\ref{sect:concrete}, we describe how to obtain the concrete
bounds given in Theorems~\ref{thm:main1} and \ref{thm:main2}.

\section{Notation and Background}
\label{sect:notation}

Throughout the paper, we identify binary strings in $\{0,1\}^n$ and
subsets of $[n]$ in the obvious way.  Analogously to the notation
$|S|$ and $\overline{S}$ for the cardinality and complement of a
subset $S \subseteq [n]$ we use $|x|$ and $\overline{x}$ for the
Hamming weight and coordinatewise complement of a string $x \in
\{0,1\}^n$.

\subsection{Submodularity}

Apart from the two definitions in the introduction, a third
characterization of submodularity is that a function $f: 2^X
\rightarrow \R^+$ is submodular if and only if
\begin{equation}
  \label{eqn:simplesubmod}
  f(S) - f(S \cup \{i\}) - f(S \cup\{j\}) + f(S \cup \{i\} \cup \{j\}) \le 0
\end{equation}
for every $S \subseteq X$, and $i,j \in X \setminus S$, $i \ne j$.  It
is straightforward to check that this condition is equivalent to the
diminishing returns property mentioned in the introduction.

\subsection{Probability}

\newcommand{\biaspk}{\{0,1\}^k_{(p)}}

For $p \in [0,1]$, we use $\biaspk$ to denote the
$k$-dimensional boolean hypercube with the $p$-biased product
distribution, i.e., if $x$ is a sample from $\biaspk$ then the
probability that the $i$'th coordinate $x_i = 1$ is $p$, independently
for each $i \in [k]$. 

We abuse notation somewhat by making no distinction between
probability distribution functions $\mu: \{0,1\}^k \rightarrow [0,1]$
and the probability space $(\{0,1\}^k, \mu)$ for such $\mu$.  Hence we
write, e.g., $\mu(x)$ for the probability of $x \in \{0,1\}^k$ under
$\mu$ and $\E_{x \sim \mu}[ f(x) ]$ for the expectation of a function
$f: \{0,1\}^k \rightarrow \R$ under $\mu$.

A distribution $\mu$ over $\{0,1\}^k$ is \emph{balanced pairwise
  independent} if every two-dimensional marginal distribution of $\mu$
is the uniform distribution, or formally, if for every $1 \le i < j
\le n$ and $b_1, b_2 \in \{0,1\}$, it holds that $$\Pr_{x \sim
  \mu}[x_i = b_1 \wedge x_j = b_2] = 1/4.$$

Recall that the support $\Supp(\mu)$ of a distribution $\mu$ over
$\{0,1\}^k$ is the set of strings with non-zero probability under
$\mu$, i.e., $\Supp(\mu) = \{\,x \in \{0,1\}^k\,:\,\mu(x) > 0\,\}$.

We conclude this section with a lemma that will be useful to us.

\begin{lemma}
  \label{lemma:symmetrybias}
  Let $f: \{0,1\}^k \rightarrow \R^+$ be a symmetric set function.
  For $t \in [0,k]$ let $a(t)$ denote the average of $f$ on strings of
  weight $x$, $a(t) = \frac{1}{{k \choose t}} \sum_{|x| = t} f(x)$.
  If $a$ is monotonely nondecreasing in $[0,k/2]$, then the maximum average of $f$ under
  any $p$-biased distribution is achieved by the uniform distribution.
  I.e.,
  $$
  \max_{p \in [0,1]} \E_{x \sim \biaspk} [f(x)] = 2^{-x} \sum_{x \in \{0,1\}} f(x)
  $$
\end{lemma}

This intuitively obvious lemma is probably well known but as we do not
know a reference we give a proof here.

\begin{proof}
  First, we note that without loss of generality we may assume that
  $f(x)$ is the indicator function of the event $k/2-d \le |x| \le
  k/2+d$ for some $d \in [0,k/2]$.  This is because any $f$ as in the
  statement of the lemma can be written as a nonnegative linear
  combination of such indicator functions for different $d$ and if the
  average of each of these indicator functions is maximized for
  $p=1/2$ then so is the average of $f$.

  Define $f_1: \{0,1\}^k \rightarrow \{0,1\}$ as the indicator
  function of the event $|x| \ge k/2-d$ and $f_2: \{0,1\}^k
  \rightarrow \{0,1\}$ as the indicator function of the event $|x| >
  k/2+d$, so that $f(x) = f_1(x) - f_2(x)$.  Let $e_j(p)$ denote the
  average of $f_j$ under the $p$-biased distribution and $e(p) =
  e_1(p)-e_2(p)$ the average of $f$ under the $p$-biased distribution.

  We will prove that $e'(p) \ge 0$ for $p \le 1/2$ (this is sufficient
  since we have $e(p) = e(1-p)$ for symmetry reasons), or in other
  words that $e_1'(p) \ge e_2'(p)$.  Now, $f_1$ and $f_2$ are
  indicator functions of monotone events and therefore $e_1'(p)$ and
  $e_2'(p)$ can be computed by the Margulis-Russo Lemma
  \cite{russo82approximate,margulis74probabilistic}:

  \begin{lemma}{(Margulis-Russo)}
    Let $f: \{0,1\}^k \rightarrow \{0,1\}$ be monotone.  For $x \in
    \{0,1\}^k$ and $i \in [k]$ let $x \setminus i$ denote $x$ with the
    $i$'th coordinate set to $0$, and let $x \cup i$ denote $x$ with
    the $i$'th coordinate set to $1$.  Then
    $$\frac{\partial}{\partial p} \E_{x \sim \biaspk}[f(x)] =
    \sum_{i=1}^k \Pr_{x \sim \biaspk}[f(x \setminus i) = 0 \wedge
    f(x \cup i) = 1].$$
  \end{lemma}

  \newcommand{\biaspkm}{\{0,1\}^{k-1}_{(p)}} 

  Applying Margulis-Russo to the monotone functions $f_1$ and $f_2$,
  and using that they depend only on $|x|$ it follows that
  (assuming without loss of generality that $d$ is such that $k/2-d$
  is an integer):
  \begin{align*}
    e_1'(p) &=  \Pr_{x \sim \biaspkm}[|x| = k/2-d-1] \cdot k &
    e_2'(p) &=  \Pr_{x \sim \biaspkm}[|x| = k/2+d] \cdot k 
  \end{align*}
  Hence to prove $e'_1(p) \ge e'_2(p)$ we have to prove that, for
  every $p \le 1/2$
  $$
  \Pr_{x \sim \biaspkm}\left[|x| = \frac{k-1}{2}-(d+\frac{1}{2})\right] \ge \Pr_{x \sim
    \biaspkm}\left[|x| = \frac{k-1}{2}+(d+\frac{1}{2})\right].
  $$
  This in turn follows immediately from $\Pr_{x \sim \biaspkm}[|x| = w] =
  {k-1 \choose w} p^w (1-p)^{k-1-w}$ since:
  $$
  \frac{\Pr_{x \sim \biaspkm}\left[|x| = \frac{k-1}{2}-(d+\frac{1}{2})\right]}{\Pr_{x \sim
    \biaspkm}\left[|x| = \frac{k-1}{2}+(d+\frac{1}{2})\right]} = 
\frac{ p^{\frac{k-1}{2}-(d+\frac{1}{2})} (1-p)^{\frac{k-1}{2}+(d+\frac{1}{2})}}{p^{\frac{k-1}{2}+(d+\frac{1}{2})}(1-p)^{\frac{k-1}{2}-(d+\frac{1}{2})}} = \left(\frac{1-p}{p}\right)^{2d+1} \ge 1.
  $$
\end{proof}

\section{Hardness from Pairwise Independence}
\label{sect:hardness}

In this section we state formally the variation of the hardness result
of \cite{austrin09approximation} that we use.  We first define the
parameters which control the inapproximability ratio that we obtain.

\begin{definition}
  Let $f: \{0,1\}^k \rightarrow \R^+$ be a submodular function.

  We define the \emph{completeness $\completeness_\mu(f)$ of $f$ with
    respect to a distribution $\mu$} over $\{0,1\}^k$ by the expected
  value of $f$ under $\mu$, i.e.,
  $$
  \completeness_\mu(f) := \E_{x \sim \mu}[f(x)]
  $$

  We define the \emph{soundness $\soundness_p(f)$ of $f$ with respect
    to bias $p$} by the expected value of $f$ under the $p$-biased
  distribution, i.e.,
  $$
  \soundness_p(f) := \E_{x \sim \biaspk}[f(x)].
  $$

  Finally, we define the \emph{soundness $s(f)$ of $f$} by its maximum
  soundness with respect to any bias, i.e.,
  $$
  \soundness(f) := \max_{p \in [0,1]} \soundness_p(f)
  $$
\end{definition}

We can now state the hardness result.

\begin{theorem}
  \label{thm:ughardness}
  Let $\mu$ be a balanced pairwise independent distribution over
  $\{0,1\}^k$.  Then for every objective function $f: \{0,1\}^k
  \rightarrow \R^+$ and $\epsilon > 0$, given a $\MaxCSP^+(f)$
  instance $F: \{0,1\}^n \rightarrow \R^+$ it is UG-hard to distinguish
  between the cases:
  \begin{description}
  \item[Yes:] There is an $S \subseteq X$ such that $F(S) \ge
    \completeness_\mu(f) - \epsilon$.
  \item[No:] For every $S \subseteq X$ it holds that $F(S) \le
    \soundness(f) + \epsilon$.
  \end{description}
\end{theorem}

The proof of Theorem~\ref{thm:ughardness} follows the proof of
\cite{austrin09approximation} almost exactly.  For the sake of
completeness, we give a bare bones proof in
Appendix~\ref{sect:hardnessproof}.

Consequently, for any submodular function $f$ and pairwise independent
distribution $\mu$ with all marginals equal, it is UG-hard to
approximate $\MaxCSP^+(f)$ to within a factor
$\soundness(f)/\completeness_\mu(f) + \epsilon$ for every $\epsilon >
0$.  Note also that the \textbf{No} case is the best possible: there
is a trivial algorithm which finds a set such that $F(S) \ge
\soundness(f)$ for every $F$, by simply letting each input be $1$ with
probability $p$ for the $p$ that maximizes $\soundness_p(f)$.

As a somewhat technical remark, we mention that
Theorem~\ref{thm:ughardness} still holds if $\mu$ is not required to
be balanced -- it suffices that all the one-dimensional marginal
probabilities $\Pr_{x \sim \mu}[x_i = 1]$ are identical, not
necessarily equal to $1/2$ as in the balanced case.  We state the
somewhat simpler form since that is sufficient to obtain our results
for submodular functions and since that makes it more similar to the
result of \cite{austrin09approximation}, which requires the
distribution $\mu$ to be balanced.

Let us then briefly discuss the difference between
Theorem~\ref{thm:ughardness} and the main result of
\cite{austrin09approximation}.  First, the result of
\cite{austrin09approximation} only applies in the more general setting
when one allows negated literals, which is why it can not be used to
obtain inapproximability for submodular functions.  On the other hand,
this more general setting allows for a stronger conclusion: in the
\textbf{No} case, \cite{austrin09approximation} achieves a soundness
of $\soundness_{1/2}(f)+\epsilon$ which in general can be much smaller
than $\soundness(f)$.  As an example, consider the case when $f:
\{0,1\}^3 \rightarrow \{0,1\}$ is the logical OR function on $3$ bits.
In this case the $\MaxCSP^+(f)$ problem is of course trivial -- the
all-ones assignment satisfies all constraints -- and $\soundness(f) =
1$, whereas $\soundness_{1/2}(f) = 7/8$.  Letting $\mu$ be the uniform
distribution on strings of odd parity (it is readily verified that
this is a balanced pairwise independent distribution) one gets
$\completeness_\mu(f) = 1$, showing that the \textsc{Max $k$-Sat}
problem is hard to approximate to within $7/8 + \epsilon$.

\section{The Construction}
\label{sect:construction}

In this section we make formal the construction outlined in
Section~\ref{sec:techniques}.

Theorem~\ref{thm:ughardness} suggests the following natural approach:
pick a pairwise independent distribution $\mu$ over $\{0,1\}^k$ and
let $\vect{1}_\mu: \{0,1\}^k \rightarrow \{0,1\}$ be the indicator
function of the support of $\mu$.  Then take $f$ to be a ``minimum
submodular upper bound'' to $\vect{1}_\mu$, by which we mean a
submodular function satisfying $f(x) \ge \vect{1}_\mu(x)$ for every
$x$ while having $\soundness(f)$ as small as possible (whereas
$\completeness_\mu(f)$ is clearly at least $1$).  Note that the
smaller the support of $\mu$, the less constrained $f$ is, meaning
that there should be more room to make $\soundness(f)$ small.

To this end, let us make the following definition.

\begin{definition}
  For a subset $\C \subseteq \{0,1\}^k$, we denote by $\SubHull(\C)$
  the optimum function $f: \{0,1\}^k \rightarrow \R^+$ of the
  following program\footnote{In the case when the optimum is not
    unique, we choose an arbitrary optimal $f$ as $\SubHull(\C)$.}:
  $$
  \begin{array}{ll}
    \text{Minimize} & \soundness(f) \\
    \text{Subject to} & \text{$f(x) \ge 1$ for every $x \in \C$} \\
    & \text{$f$ is submodular}
  \end{array}
  $$
  In addition, we write $\SubHull_p(\C)$ for the optimal $f$ when the
  objective to be minimized is changed to $\soundness_p(f)$ instead of
  $\soundness(f)$.  Analogously, we define $\SubHullSym(\C)$ and
  $\SubHullSym_p(\C)$ as the optimal $f$ with the additional
  restriction that $f$ is symmetric.
\end{definition}

While the objective function $\soundness(f)$ is not linear (or even
convex), it turns out that for the $\C$'s that we are interested in,
$\SubHull(\C)$ is actually quite well approximated by
$\SubHull_{1/2}(\C)$, i.e., we simply minimize $\sum_x f(x)$ (in fact,
we even believe that for our $\C$'s $\SubHull_{1/2}(\C)$ gives the
exact optimum for $\SubHull(\C)$, though we have not attempted to
prove it).  The advantage of considering $\SubHull_{1/2}(\C)$ is of course
that it is given by a linear program, which gives us a reasonably
efficient way of finding it.  Armed with this definition, let us now
describe the constructions we use.

\subsection{The Asymmetric Case}

The family of pairwise independent distributions $\mu$ that we
consider is a standard construction based on the Hadamard code.  Fix a
parameter $l > 0$ and let $k = 2^l-1$.  We identify the set of
coordinates $[k]$ with the set of non-empty subsets of $[l]$, in some
arbitrary way.  A string $x$ from the distribution $\mu$ is sampled as
follows: pick a uniformly random string $y \in \{0,1\}^l$ and
defining, for each $\emptyset \ne T \subseteq [l]$, the coordinate
$x_T = \bigoplus_{i \in T} y_i$.

This construction already has an issue: since the all-zeros string
$\vect{0}$ is in the support of the distribution, any submodular upper
bound to $\vect{1}_\mu$ must have $f(\vect{0}) \ge 1$, implying that
$\soundness_0(f) = 1$.  To fix this, we simply ignore $\vect{0}$ when
constructing $f$.  Formally, let $\C_l = \Supp(\mu) \setminus
\{\vect{0}\} \subseteq \{0,1\}^k$ be the $2^l-1$ strings in the
support of $\mu$ except $\vect{0}$.  Now we would like to take our
submodular function $f$ to be $\SubHull(\C)$, but we instead take it
to be $\SubHull_{1/2}(\C)$, as this function is much more easily
computed.

\begin{definition}
  For a parameter $l > 0$, let $k = 2^l-1$ and take $\C_l \subseteq
  \{0,1\}^k$ as above.  We define $f_l = \SubHull_{1/2}(\C_l)$.
\end{definition}

Note that using only $\C_l$ instead of the entire support costs us a
little in that the completeness is now reduced from $1$ to
$\completeness_\mu(f_l) \ge 1-2^{-l}$, but one can hope (and it indeed
turns out that this is the case) that this loss is compensated by a
greater improvement in soundness.

Also, we stress that $\soundness(f_l)$ is typically \emph{not} given
by the average $\soundness_{1/2}(f_l)$ (which is the quantity actually
minimized by $f_l$).  Indeed, the points in $\C_l$ all have Hamming
weight $(k+1)/2$ and this is also where $f_l$ is typically the
largest.  This causes $\soundness(f)$ to be achieved by the $p$-biased
distribution for some $p$ slightly larger than $1/2$. 

An obvious question to ask is whether using $\SubHull(\C_l)$ would give a
better result than using $\SubHull_{1/2}(\C_l)$.  For the values of $l$
that we have been able to handle, it appears that the answer to this
question is negative: computing $\SubHull_{p}(\C_l)$ for a $p$ that
approximately maximizes $\soundness_p(f_l)$ gives $f_l$, indicating
that we in fact have $f_l = \SubHull(\C_l)$.

\subsection{Symmetric Functions}

One way of constructing symmetric functions would be to use the exact
same construction as above but taking $\SubHullSym(\C_l)$ rather than
$\SubHull(\C_l)$.  However, that is somewhat wasteful, and we achieve
better results by also taking symmetry into account when constructing
the family of strings $\C$.

Thus, we alter the above construction as follows: rather than
identifying the coordinates with all non-empty subsets of $[l]$, we
identify them with all subsets of $[l]$ of odd cardinality.  In other
words, we take $k = 2^{l-1}$ and associate $[k]$ with all $T \subseteq
[l]$ such that $|T|$ is odd.  The resulting distribution $\mu$ is
symmetric in the sense that if $x$ is in the support then so is
$\overline{x}$.

In this case, both the all-zeros string $\vect{0}$ and the all-ones
string $\vect{1}$ are in the support which is not acceptable for the
same reason as above.  Hence, we construct a submodular function by
taking $\Csym_l = \Supp(\mu) \setminus \{\vect{0}, \vect{1}\}$ (note
that $|\Csym_l| = 2^l-2$).

\begin{definition}
  For a parameter $l > 0$, let $k = 2^{l-1}$ and take $\Csym_l
  \subseteq \{0,1\}^k$ as above.  We define $\fsym_l =
  \SubHullSym_{1/2}(\Csym_l)$.
\end{definition}

In this case, since we removed $2$ out of the $2^l$ points of the
support of $\mu$ to construct $\Csym_l$, we have that
$\completeness_\mu(\fsym_l) \ge 1 - 2^{1-l}$.  

An salient feature of $\fsym_l$ is that all strings of $\Csym_l$ have
Hamming weight exactly $k/2$.  By Lemma~\ref{lemma:symmetrybias}, this
causes $\soundness_p(\fsym_l)$ to be maximized by $p = 1/2$ (the
monotonicity of the function $a$ in Lemma~\ref{lemma:symmetrybias} is
not immediately clear).  This means that in the symmetric case, using
$\SubHullSym_{1/2}(\Csym_l)$ rather than $\SubHullSym(\Csym_l)$ is
provably without loss of generality.

\section{Concrete Bounds}
\label{sect:concrete}

Unfortunately, understanding the behaviour of the two families of
functions $f_l$ and $\fsym_l$ (or even just their soundnesses) for
large $l$ appears difficult.  There seems to be two conflicting forces
at work: on the one hand, $\C_l$ only has $2^l-1=k$ points so even
though $f_l$ is forced to be large on these there may still be plenty
of room to make it small elsewhere.  But on the other hand, since
$\C_l$ is a good code the elements of $\C_l$ are very pread out (their
pairwise Hamming distances are roughly $k/2$), which together with the
submodularity condition appears to force $f_l$ to be large.

In this section we study $f_l$ for small $l$, obtaining our hardness
results.  As discussed towards the end of the section, there are
indications that the inapproximability given by $f_l$ actually becomes
worse for large $l$ and that our results are the best possible for
this family of functions, but we do not yet know whether these
indications are correct.

\subsection{Symmetric Functions}

We start with the symmetric functions, as these are somewhat nicer
than the asymmetric ones in that their symmetry turn out to cause
$s(\fsym_l)$ to be achieved by $p=1/2$, i.e., $\soundness(\fsym_l)$
simply equals the average of $\fsym_l$.  Table~\ref{table:fsymsum}
gives a summary of the completeness, soundness, and inapproximability
obtained by $\fsym_l$ for $l \in \{3,4,5\}$.  We now describe these
functions in a more detail.

\begin{table}
  \centering
  \begin{tabular}{|c|c|c|ccc|}
    \hline
    $l$ & $\completeness$ & $\soundness(\fsym_l)$ & \multicolumn{3}{|c|}{Inapproximability $\soundness/\completeness$} \\
    \hline
    $3$ & 
    $3/4$ & $5/8$ & $5/6$ & $<$ & $0.8334$ 
    \\
    $4$ &
    $7/8$ & $43/64$ & $43/56$ & $<$ & $0.7679$ 
    \\
    $5$ & 
    $15/16$ & $709/1024$ & $709/960$ & $<$ & $0.7386$
    \\
    \hline
  \end{tabular}
  \caption{Behaviour of $\fsym_l$ for small $l$.}
  \label{table:fsymsum}
\end{table}

As a warmup, let us first describe the quite simple function $\fsym_4:
2^{[8]} \rightarrow [0,1]$ (we leave the even easier function
$\fsym_3$ to the interested reader).  Its definition is as follows:
$$
\fsym_4(S) = \left\{
  \begin{array}{ll}
    f(\overline{S})  & \text{if $|S| > 4$} \\
    |S|/4            & \text{if $|S| < 4$} \\
    1                & \text{if $|S| = 4$ and $S$ is in $\Csym_4$}\\
    3/4              & \text{otherwise}
  \end{array}
\right..
$$
That $\fsym_4(S)$ is submodular is easily verified.  It is also easy
to check that Lemma~\ref{lemma:symmetrybias} applies and therefore we
have that $\soundness(\fsym_4) = \soundness_{1/2}(\fsym_4)$, which is
straightforward to compute (note that $|\Csym_4| = 14$):
$$
\soundness_{1/2}(\fsym_4) = 2^{-8} \left(2 {8 \choose 1} \cdot \frac{1}{4} + 2 {8 \choose 2} \cdot \frac{2}{4} + 2 {8 \choose 3} \cdot \frac{3}{4} + 14 \cdot 1 + \left(\!{8 \choose 4}-14\right) \cdot \frac{3}{4} \right) = \frac{43}{64}
$$

Let us then move on to the next function $\fsym_5: 2^{[16]}
\rightarrow [0,1]$, giving an inapproximability of $0.7386$.  It turns
out that one can take $\fsym_5(S)$ to be a function of two simple
properties of $S$, namely its cardinality $|S|$, and the distance from
$S$ to $\Csym_5$.  Specifically, for $|S| \le 8$ let us define the
\emph{number of errors} $e(S)$ as the minimum number of elements that
must be removed from $S$ to get a subset of some set in $\Csym_5$.
Formally
$$
e(S) = \min_{C \in \Csym_5} |S \setminus C|,
$$
or equivalently, $d(S,\Csym_5) = 8-|S|+2e(S)$, where $d(S, \Csym_5)$
is the Hamming distance from the binary string corresponding to $S$ to
the nearest element in $\Csym_5$.  Table~\ref{table:fsym5} gives the
values of $\fsym_5$ for all $|S| \le 8$, and for $|S| > 8$ the value
of $\fsym_5(S)$ is given by $\fsym_5(\overline{S})$.  Note that, for
sets with $e(S) = 0$, i.e., no errors, $\fsym_5(S)$ is simply $|S|/8$,
which is what one would expect.  However, for sets with errors,
$\fsym_5(S)$ has a more complicated behaviour and it is far from clear
how this generalizes to larger $l$.

\begin{table}
  \centering
  \begin{tabular}{|c|ccccccccc|}
    \hline
    & \multicolumn{9}{|c|}{$|S|$} \\
    $e(S)$ & $0$ & $1$ & $2$ & $3$ & $4$ & $5$ & $6$ & $7$ & $8$ \\
    \hline
    0 & 0 & 1/8  & 2/8 & 3/8 & 4/8 & 5/8 & 6/8 & 7/8 & 1 \\
    1 & --& --& --& -- & -- & 19/32 & 22/32 & 24/32 & 26/32 \\
    2 & --& --& --& -- & -- & -- & 20/32 & 23/32 & 24/32 \\
    \hline
  \end{tabular}
  \caption{Description of $\fsym_5(S)$ as a function of $|S|$ and $e(S)$ for $|S| \le 8$.}
  \label{table:fsym5}
\end{table}

Veryfing that $\fsym_5$ is indeed submodular is not as straightforward
as with $\fsym_4$.  We have not attempted to construct a shorter proof
of this than simply checking condition (\ref{eqn:simplesubmod}) for
every $S$, $i$ and $j$, a task which is of course best suited for a
computer program (which is straightforward to write and runs in a few
seconds).

A computer program is also the best way to compute the soundness
$\soundness(\fsym_5)$.  It is almost obvious from inspection of
Table~\ref{table:fsym5} that $\fsym_5$ satisfies the monotonicity
condition of Lemma~\ref{lemma:symmetrybias} (the only possible source
of failure is that the table only implies that the average of
$\fsym_5$ on sets of size $6$ is between $20/32$ and $24/32$, and that
the average on sets of size $7$ is between $23/32$ and $28/32$).  It
turns out that the conditions of Lemma~\ref{lemma:symmetrybias} are
indeed satisfied and that the average of $\fsym_5$ is
$\soundness_{1/2}(\fsym_5) = 709/1024$.

Concluding this discussion on $\fsym_l$, it is tempting to speculate
on its behaviour for larger $l$.  We have made a computation of
$\fsym_6: 2^{[32]} \rightarrow [0,1]$, under the assumption that
$\fsym_6(S)$ only depends on $|S|$ and the multiset of distances to
every point of the support of $\Csym_6$.  Under this assumption, our
computations indicate that $\soundness(\fsym_6) \approx 0.7031$ giving
an inapproximatibility of $\soundness(\fsym_6)/(31/32) \approx
0.7258$, improving upon $\fsym_5$.  However, as these computations
took a few days they are quite cumbersome to verify (and we have not
even made a careful verification of them ourselves) and therefore we
do not claim this stronger hardness as a theorem.

\subsection{Asymmetric Functions}

We now return our focus to the asymmetric case.
Table~\ref{table:fsum} describes the hardness ratios obtained from
$f_l$ for the cases $l = 3$ and $l = 4$.

\begin{table}
  \centering
  \begin{tabular}{|c|c|c|c|}
    \hline
    $l$ & $\completeness$ & $\soundness(f_l)$ & Inapproximability $\soundness/\completeness$ \\
    \hline
    $3$ & 
    $7/8$ & $< 0.6275$ & $<0.7172$ 
    \\
    $4$ &
    $15/16$ & $< 0.6508$ & $< 0.6942$ 
    \\
    \hline
  \end{tabular}
  \caption{Behaviour of $f_l$ for small $l$.}
  \label{table:fsum}
\end{table}

We begin with the description of the function $f_3: 2^{[7]}
\rightarrow [0,1]$.  Similarly to the definition $e(S)$ used in the
description of $\fsym_5$, let us say that $S \subseteq [7]$ \emph{has
  no errors} if it is a subset or a superset of some $C \in \C_3$.  In
other words, if $|S| < 4$ it has no errors if it can be transformed to
a set in $\C_3$ by adding some elements, and if $|S| > 4$ it is has no
errors if it can be transformed to a codeword by removing some
elements.  The function $f_3$ is as follows:
$$
f_3(S) = \left\{
  \begin{array}{ll}
    |S|/4 & \text{if $|S| \le 4$ and has no errors} \\
    (7-|S|)/3 & \text{if $|S| > 4$ and has no errors}\\
    11/24 & \text{if $|S| = 3$ and has errors}\\
    17/24 & \text{if $|S| = 4$ and has errors}
  \end{array}\right.
$$
As with $\fsym_5$, it is not completely obvious that $f_3$ satisfies
the submodularity condition and there are a few cases to verify, best
left to a computer program.

The average of $f_3$ is $637/1024 \approx 0.622$.  However, since
$f_3$ takes on its largest values at sets of size $(k+1)/2 = 4$, the
$p$-biased average is larger than this for some $p > 1/2$.  It turns
out that $\soundness(f_4)$ is obtained by the $p$-biased distribution
for $p \approx 0.542404$, giving $\soundness(f_4) \approx 0.627434 <
0.6275$.



We are left with the description of $f_4: 2^{[15]} \rightarrow [0,1]$,
which is also the most complicated function yet.  One might hope that
$f_4$ shares the simple structure of the previous functions -- that it
depends only on $|S|$ and the distance of $S$ to the nearest $C \in
\C_4$.  However, the best function under this assumption turns out to
give a worse result than $f_3$.  Instead, $f_4$ depends on $|S|$ and
the multiset of distances to all elements of $\C_4$.  

To describe $f_4$, define for $S \subseteq [15]$ the multiset
$\mathcal{D}(S)$ as the multiset of distances to all the $15$ strings
in $\C_4$.  For instance, for $S = \emptyset$, $\mathcal{D}(S)$
consists of the number $8$ repeated $15$ times, reflecting the fact
that all strings of $\C_4$ have weight $8$, and for $S \in \C_4$ we
have that $\mathcal{D}(S)$ consists of the number $8$ repeated $14$
times, together with a single $0$, because the distance between any
pair of strings in $\C_4$ is $8$.

  \begin{table}
    \centering
    \begin{tabular}{|c|l|r|c|}
      \hline
      $|S|$ & $\mathcal{D(S)}$ & $\#S$ & $448 \cdot f_4(S)$ \\
      
      \hline 
      $0$ & $8^{15}$ & $1$ & $0$ \\ 
      \hline 
      $1$ & $7^{8}9^{7}$ & $15$ & $56$ \\ 
      \hline 
      $2$ & $6^{4}8^{8}10^{3}$ & $105$ & $112$ \\ 
      \hline 
      $3$ & $5^{2}7^{6}9^{6}11^{1}$ & $420$ & $168$ \\ 
      $3$ & $7^{12}11^{3}$ & $35$ & $138$ \\ 
      \hline 
      $4$ & $4^{2}8^{12}12^{1}$ & $105$ & $224$ \\ 
      $4$ & $4^{1}6^{4}8^{6}10^{4}$ & $840$ & $224$ \\ 
      $4$ & $6^{6}8^{6}10^{2}12^{1}$ & $420$ & $194$ \\ 
      \hline 
      $5$ & $3^{1}5^{1}7^{6}9^{6}11^{1}$ & $840$ & $280$ \\ 
      $5$ & $5^{5}9^{10}$ & $168$ & $280$ \\ 
      $5$ & $5^{3}7^{6}9^{4}11^{2}$ & $1680$ & $250$ \\ 
      $5$ & $5^{2}7^{8}9^{4}13^{1}$ & $315$ & $220$ \\ 
      \hline 
      $6$ & $2^{1}6^{3}8^{8}10^{3}$ & $420$ & $336$ \\ 
      $6$ & $4^{2}6^{3}8^{6}10^{4}$ & $1680$ & $306$ \\ 
      $6$ & $4^{1}6^{5}8^{6}10^{2}12^{1}$ & $2520$ & $276$ \\ 
      $6$ & $6^{9}10^{6}$ & $280$ & $276$ \\ 
      $6$ & $6^{6}8^{8}14^{1}$ & $105$ & $216$ \\ 
      \hline 
      $7$ & $1^{1}7^{7}9^{7}$ & $120$ & $392$ \\ 
      $7$ & $3^{1}5^{2}7^{5}9^{6}11^{1}$ & $2520$ & $332$ \\ 
      $7$ & $3^{1}7^{11}11^{3}$ & $420$ & $302$ \\ 
      $7$ & $5^{4}7^{5}9^{4}11^{2}$ & $2520$ & $302$ \\ 
      $7$ & $5^{3}7^{7}9^{4}13^{1}$ & $840$ & $272$ \\ 
      $7$ & $7^{14}15^{1}$ & $15$ & $197$ \\ 
      \hline
    \end{tabular}
    $\quad\quad\quad$
    \begin{tabular}{|c|l|r|c|}
      \hline
      $|S|$ & $\mathcal{D(S)}$ & $\#S$ & $448 \cdot f_4(S)$ \\
      \hline 
      $8$ & $0^{1}8^{14}$ & $15$ & $448$ \\ 
      $8$ & $2^{1}6^{4}8^{7}10^{3}$ & $840$ & $358$ \\ 
      $8$ & $4^{3}8^{11}12^{1}$ & $420$ & $328$ \\ 
      $8$ & $4^{2}6^{4}8^{5}10^{4}$ & $2520$ & $328$ \\ 
      $8$ & $4^{1}6^{6}8^{5}10^{2}12^{1}$ & $2520$ & $298$ \\ 
      $8$ & $6^{7}8^{7}14^{1}$ & $120$ & $253$ \\ 
      \hline 
      $9$ & $1^{1}7^{8}9^{6}$ & $105$ & $384$ \\ 
      $9$ & $3^{1}5^{2}7^{6}9^{5}11^{1}$ & $2520$ & $324$ \\ 
      $9$ & $5^{6}9^{9}$ & $280$ & $324$ \\ 
      $9$ & $5^{4}7^{6}9^{3}11^{2}$ & $1680$ & $294$ \\ 
      $9$ & $5^{3}7^{8}9^{3}13^{1}$ & $420$ & $279$ \\ 
      \hline 
      $10$ & $2^{1}6^{4}8^{8}10^{2}$ & $315$ & $320$ \\ 
      $10$ & $4^{2}6^{4}8^{6}10^{3}$ & $1680$ & $290$ \\ 
      $10$ & $4^{1}6^{6}8^{6}10^{1}12^{1}$ & $840$ & $275$ \\ 
      $10$ & $6^{10}10^{5}$ & $168$ & $260$ \\ 
      \hline 
      $11$ & $3^{1}5^{2}7^{6}9^{6}$ & $420$ & $256$ \\ 
      $11$ & $3^{1}7^{12}11^{2}$ & $105$ & $256$ \\ 
      $11$ & $5^{4}7^{6}9^{4}11^{1}$ & $840$ & $241$ \\ 
      \hline 
      $12$ & $4^{3}8^{12}$ & $35$ & $192$ \\ 
      $12$ & $4^{1}6^{6}8^{6}10^{2}$ & $420$ & $192$ \\ 
      \hline 
      $13$ & $5^{3}7^{8}9^{4}$ & $105$ & $128$ \\ 
      \hline 
      $14$ & $6^{7}8^{8}$ & $15$ & $64$ \\ 
      \hline 
      $15$ & $7^{15}$ & $1$ & $0$ \\ 
      \hline
    \end{tabular}
    \caption{Description of $f_4$}
    \label{tbl:f4}
  \end{table}

Table~\ref{tbl:f4} 
describes the behaviour of $f_4(S)$ as a function of $|S|$ and
$\mathcal{D}(S)$.\footnote{It is not necessary to include $|S|$ as
  it is uniquely determined by $\mathcal{D}(S)$, but we find that
  explicitly including $|S|$ makes the table somewhat less obscure.}
In the table $\mathcal{D}(S)$ is described by a string of the form
$d_1^{m_1}d_2^{m_2} \ldots$, with $d_1 < d_2 < \ldots$ and $\sum m_i
= 15$, indicating that $m_1$ strings of $\C_4$ are at distance $d_1$
from $S$, that $m_2$ strings are at distance $d_2$, and so on.
Thus, for $S = \emptyset$ the description of $\mathcal{D}(S)$ is
``$8^{15}$'', and for $S \in \C_4$ the description of $\mathcal{D}(S)$ is
``$0^18^{14}$''.

The $\#S$ column of Table~\ref{tbl:f4} gives the total number of $S
\subseteq [15]$ having this particular value of $(|S|,
\mathcal{D}(S))$, and the last column gives the actual value of $f_4$,
multiplied by $448$ to make all values integers.

Again, checking that $f_4$ is submodular is a tedious task best suited
for a computer.  The average of $f_4$ is $9519345/(448 \cdot 2^{15}) \approx
0.6485$, but, as with $f_3$, $\soundness(f_4)$ is somewhat larger than
this.  It turns out that the $p$ maximizing $\soundness_p(f_4)$ is
roughly $p \approx 0.526613$, and that $\soundness(f_4) \approx
0.650754 < 0.6508$.

Finally, we mention that as in the symmetric case, we have made a
computation of the next function, $f_5$, again under the assumption
that it depends only on the multiset of distances to the codewords.
Under this assumption it turns out that $\soundness_{1/2}(f_5) \approx
0.6743$, meaning that the inapproximability obtained can not be better
than $\soundness_{1/2}(f_5) / (31/32) \approx 0.6961$ which is worse
than the inapproximability obtained from $f_4$.

\section{Acknowledgments}
We are grateful to Jan Vondr\'{a}k for stimulating discussions.

\bibliographystyle{abbrv}
\bibliography{references}

\begin{appendix}

  \section{Proof of Theorem~\ref{thm:ughardness}}
  \label{sect:hardnessproof}

  To prove Theorem~\ref{thm:ughardness} we only give a
  \emph{dictatorship test} with certain properties.  The method of
  translating such a test into a hardness result under the UGC, going
  back to the results of Khot et al.\ \cite{khot07optimal} for
  $\textsc{Max Cut}$ is by now quite standard (see e.g.\
  \cite{raghavendra08optimal}).  

  \subsection{Background: Polynomials, Quasirandomness and Correlation Bounds}

  To set up the dictatorship test we need to mention some background
  material.

  A function $F: \{0,1\}^n \rightarrow \R$ is said to a be a
  \emph{dictator} if $G(x) = x_i$ for some $i \in [n]$, i.e., $G$ simply
  returns the $i$'th coordinate.

  Now, any function $F: \{0,1\}^n \rightarrow \R$ can be written
  uniquely as a multilinear polynomial $F(x) = \sum_{S \subseteq [n]}
  c_S x^S$ for some set of coefficients $c_S$, where $x^S := \prod_{i
    \in S} x_i$.  With this view there is an obvious extension of the
  domain of $F$ to $[0,1]^n$ (or even $\R^n$, but we shall only be
  interested in $[0,1]^n$).

  We say that such a polynomial is \emph{$(d, \tau)$-quasirandom} if for
  every $i \in [n]$ it holds that
  $$
  \sum_{\substack{i \in S \subseteq [n] \\ |S| \le d}} c_S^2 \le \tau.
  $$
  Note that a dictator is in some sense the extreme opposite of a
  $(d,\tau)$-quasirandom function as a dictator is not even
  $(1,\tau)$-quasirandom for $\tau < 1$.

  The main tool to obtain the soundness is the following ``noise
  correlation bound'' result of Mossel \cite{mossel08gaussian}
  (Theorem 6.6 and Lemma 6.9), which we state here in a simplified
  form in order to keep the amount of background necessary to a
  minimum.

  \begin{theorem}
    \label{thm:correlation_bound}
    Let $\epsilon > 0$ and let $\mu$ be a balanced pairwise independent probability
    distribution over $\{0,1\}^k$ such that $\mu(x) > 0$ for every $x
    \in \{0,1\}^k$.  Then there exists $d, \tau
    > 0$ such that the following holds for all $n$.  
    
    Let $F_1, \ldots, F_k: \{0,1\}^n \rightarrow [0,1]$ be
    $(d,\tau)$-quasirandom functions.  Then
    \begin{eqnarray*}
      \Bigg| \E_{w_1,\ldots,w_n}\left[ \prod_{i=1}^{k} F_i(w_{1,i}, \ldots, w_{n,i}) \right] - \prod_{i=1}^k \E[F_i] \Bigg| \le \epsilon,
    \end{eqnarray*}
    where $w_1, \ldots, w_n \in \{0,1\}^k$ are drawn independently from
    $\mu$ and $w_{i,j} \in \{0,1\}$ denotes the $j$th coordinate of
    $w_i$.
  \end{theorem}

  \subsection{Dictatorship Test}

  We now give the dictatorship test, which by the standard conversion
  from dictatorship tests to hardness implies
  Theorem~\ref{thm:ughardness}.  In the dictatorship test, the
  function $f: [0,1]^k \rightarrow [0,1]$ has the same role as the
  function $f: \{0,1\}^k \rightarrow \R^+$ in
  Theorem~\ref{thm:ughardness} -- as mentioned in the previous section
  we can take the unique multilinear extension to make the domain the
  entire $[0,1]^k$, and the range can be taken to be $[0,1]$ without
  loss of generality by simply scaling the function down.

  \begin{theorem}
    For every $\epsilon$ there are $d, \tau > 0$ such that the
    following holds.  Let $f: [0,1]^k \rightarrow [0,1]$ and $\mu$ be
    a balanced pairwise independent distribution over $\{0,1\}^k$.
    There is a dictatorship test $\mathcal{A}$, which when run on a
    function $F: \{0,1\}^n \rightarrow [0,1]$ has the following
    properties:
    \begin{enumerate}
    \item $\mathcal{A}$ queries $F$ in $k$ positions $x_1, \ldots, x_k
      \in \{0,1\}^n$ and then accepts with probability $f(F(x_1), \ldots,
      F(x_k))$.
    \item If $F$ is a dictator then $\mathcal{A}$ accepts with
      probability at least $\completeness_\mu(f) - \epsilon$.
    \item If $F$ is $(d,\tau)$-quasirandom then $\mathcal{A}$
      accepts with probability at most $\soundness(f) + \epsilon$.
    \end{enumerate}
  \end{theorem}  

  \begin{proof}
    Let $\mu'$ be the distribution over $\{0,1\}^k$ defined by
    $$
    \mu' = (1-\epsilon) \mu + \epsilon \mathcal{U},
    $$
    where $\mathcal{U}$ denotes the uniform distribution (in other
    words, a sample from $\mu'$ is obtained by sampling from $\mu$
    with probability $1-\epsilon$ and otherwise, with probability $\epsilon$,
    taking a uniformly random element of $\{0,1\}^k$).  Note that
    $\mu'$ is also balanced pairwise independent, and more importantly
    it satisfies $\mu'(x) > 0$ for all $x \in \{0,1\}^k$ which will
    allow us to apply Theorem~\ref{thm:correlation_bound}.

    Now the test $\mathcal{A}$ is as follows:
    \begin{itemize}
    \item Pick a random $k$-by-$n$ matrix $X$ over $\{0,1\}$ by letting each
      column be a sample from $\mu'$, independently.
    \item Let $x_1, \ldots, x_k \in \{0,1\}^n$ be the rows of $X$ and
      let $F(X) = (F(x_1), \ldots, F(x_k)) \in \{0,1\}^k$ be the values
      of $F$ on these $k$ points.
    \item Accept with probability $f(F(X))$.
    \end{itemize}

    The first property of $\mathcal{A}$ is clear from its definition.
    For the completeness property, note that if $F$ is a dictator then
    $F(X) \in \{0,1\}^k$ is just some column of $X$
    and therefore distributed according to $\mu'$, so that
    $$
    \E[f(F(X))] = \E_{x \sim \mu'}[f(x)] =
    (1-\epsilon)\E_{x \sim \mu}[f(x)] + \epsilon \E_{x \sim \mathcal{U}}[f(x)]
    \ge \E_{x \sim \mu}[f(x)] - \epsilon = \completeness_\mu(f) - \epsilon.
    $$

    We now turn to the soundness property of $\mathcal{A}$.  Let
    $\epsilon' = \epsilon / 2^k$ and let $d$ and $\eta$ be given by
    Theorem~\ref{thm:correlation_bound} with parameter $\epsilon'$ and
    the distribution $\mu'$.  

    Now consider the multilinear expansion $f(x) = \sum_{S \subseteq
      [k]} c_S x_S$ of $f$ and let us analyze the expectation of
    $f(F(X))$ term by term.  If $F$ is $(d,\tau)$-quasirandom then by
    Theorem~\ref{thm:correlation_bound} (letting $F_i = F$ for $i \in
    S$ and letting $F_i$ be the constant one function for $i \not \in S$) we have
    $$
    \left|\E[\prod_{i \in S} F(x_i)] - \prod_{i \in S} \E[F]\right| \le \epsilon'.
    $$
    Let $p = \E[F]$ be the bias of the function $F$.  Then, $\prod_{i
      \in S} \E[F] = p^{|S|}$ equals the expectation of $x^S$ under the
    $p$-biased distribution.  Summing over all $S$ we obtain
    $$
    \E[f(F(X))] \le \sum_{S \subseteq [k]} c_S \E_{x \sim
      \biaspk}[x^S] + 2^k \epsilon' = \E_{x \sim \biaspk}[f(x)] +
    \epsilon = \soundness_p(f) + \epsilon \le \soundness(f) +
    \epsilon,
    $$
    giving the desired soundness property.
  \end{proof}

\end{appendix}

\end{document}